\documentclass[a4paper,UKenglish,cleveref,autoref]{lipics-v2019}
\usepackage{microtype}

\title{An algorithmic weakening of the Erd\H{o}s-Hajnal conjecture}

\titlerunning{An algorithmic weakening of the Erd\H{o}s-Hajnal conjecture} 

\author{\'{E}douard Bonnet}{Univ Lyon, CNRS, ENS de Lyon, Universit\'{e} Claude Bernard Lyon 1, LIP UMR5668, France}{}{}{}
\author{St\'{e}phan Thomass\'{e}}{Univ Lyon, CNRS, ENS de Lyon, Universit\'{e} Claude Bernard Lyon 1, LIP UMR5668, France}{}{}{}
\author{Xuan Thang Tran}{Univ Lyon, CNRS, ENS de Lyon, Universit\'{e} Claude Bernard Lyon 1, LIP UMR5668, France}{}{}{}
\author{R\'{e}mi Watrigant}{Univ Lyon, CNRS, ENS de Lyon, Universit\'{e} Claude Bernard Lyon 1, LIP UMR5668, France}{}{}{}
\authorrunning{\'E. Bonnet, S. Thomass\'{e}, X. T. Tran, R. Watrigant}

\Copyright{\'E. Bonnet, S. Thomass\'{e}, X. T. Tran, R. Watrigant}



\category{}

\relatedversion{}

\supplement{}

\funding{}

\acknowledgements{}

\EventEditors{}
\EventNoEds{1}
\EventLongTitle{}
\EventShortTitle{}
\EventAcronym{}
\EventYear{}
\EventDate{}
\EventLocation{}
\EventLogo{}
\SeriesVolume{}
\ArticleNo{}
\nolinenumbers 
\hideLIPIcs  

\usepackage{xspace}
\usepackage[table]{xcolor}
\usepackage{tikz}
\usepackage{array}
\usepackage[shortlabels]{enumitem}
\usepackage{bbm}
\usetikzlibrary{fit}
\usetikzlibrary{calc}
\usetikzlibrary{shapes}
\usepackage[noend]{algpseudocode}
\usepackage[]{algorithmicx}

\newcommand{\mis}{\textsc{MIS}\xspace}
\newcommand{\smis}{\textsc{MIS}\xspace}

\newcommand{\nice}{locally easy\xspace}
\newcommand{\tnice}{$t$-locally easy\xspace}

\newcommand{\pnice}[1]{${#1}$-locally easy\xspace}

\theoremstyle{plain}

\newtheorem{conjecture}[theorem]{Conjecture}
\newtheorem{observation}[theorem]{Observation}

\newcommand{\whp}{\emph{w.h.p.}\xspace}

\begin{document}

\maketitle

\begin{abstract}
We study the approximability of the \textsc{Maximum Independent Set} (\textsc{MIS}) problem in $H$-free graphs (that is, graphs which do not admit $H$ as an induced subgraph). As one motivation we investigate the following conjecture: for every fixed graph $H$, there exists a constant $\delta > 0$ such that \textsc{MIS} can be $n^{1 - \delta}$-approximated in $H$-free graphs, where $n$ denotes the number of vertices of the input graph.
We first prove that a constructive version of the celebrated Erd\H{o}s-Hajnal conjecture implies ours. 
We then prove that the set of graphs $H$ satisfying our conjecture is closed under the so-called graph substitution. This, together with the known polynomial-time algorithms for \textsc{MIS} in $H$-free graphs (e.g. $P_6$-free and fork-free graphs), implies that our conjecture holds for many graphs $H$ for which the Erd\H{o}s-Hajnal conjecture is still open.
We then focus on improving the constant $\delta$ for some graph classes: we prove that the classical \textsc{Local Search} algorithm provides an $OPT^{1-\frac{1}{t}}$-approximation in $K_{t, t}$-free graphs (hence a $\sqrt{OPT}$-approximation in $C_4$-free graphs), and, while there is a simple $\sqrt{n}$-approximation in triangle-free graphs, it cannot be improved to $n^{\frac{1}{4}-\varepsilon}$ for any $\varepsilon > 0$ unless $NP \subseteq BPP$. More generally, we show that there is a constant $c$ such that \textsc{MIS} in graphs of girth~$\gamma$ cannot be $n^{\frac{c}{\gamma}}$-approximated. Up to a constant factor in the exponent, this matches the ratio of a known approximation algorithm by Monien and Speckenmeyer, and by Murphy. To the best of our knowledge, this is the first strong (i.e., $\Omega(n^\delta)$ for some $\delta > 0$) inapproximability result for \textsc{Maximum Independent Set} in a proper hereditary class.
\end{abstract}

\section{Introduction}\label{sec:intro}

An \emph{independent set} of a (simple, undirected) graph is a set of pairwise non-adjacent vertices.
Independent sets have been central in various research topics, both in algorithmic and structural graph theory. 
In structural graph theory, independent sets (and their complements, cliques) are at the core of several celebrated results, such as K\H{o}nig's theorem, Ramsey's theorem, or Turan's theorem \cite{BondyMurty08}, to name only a few.
Finding an independent set of maximum cardinality (called the \textsc{Maximum Independent Set} problem, or \smis for short) is a fundamental intractable optimization problem.
Indeed, it is NP-hard to solve \cite{GJ79}, but also to approximate within ratio $n^{1-\varepsilon}$ for any $\varepsilon > 0$ \cite{Hastad96,Zuckerman07}, where $n$ denotes the number of vertices of the input graph.
On the positive side, \smis becomes tractable when restricted to some specific graph classes: It is polynomial-time solvable in bipartite graphs and more generally in perfect graphs \cite{Grotschel88}, admits a PTAS in planar graphs~\cite{Baker94} and in more general geometric graph classes such as pseudo-disk graphs~\cite{Chan03}, bounded genus or $H$-minor-free graphs~\cite{Demaine05}. 
Notice that all the aforementioned graph classes are closed under taking induced subgraphs.
We call \emph{hereditary} such a class, and add the qualificative \emph{proper} if it is not the class of all graphs.
A hereditary class can be defined by a (possibly infinite) set of forbidden induced subgraphs.
A potentially unifying framework is to consider the complexity of \mis in $H$-free graphs (i.e., graphs without induced copy of $H$) and $\mathcal H$-free graphs (i.e., graphs without induced copy of any $H \in \mathcal H$).
However, a classical reduction \cite{Poljak74,Alekseev82} consisting of subdividing every edge of a given graph $G$ a fixed even number of times $2c$ leads to a graph $G'$ such that $\alpha(G') = \alpha(G)+c |E(G)|$ (where $\alpha(.)$ denotes the size of a maximum independent set of a graph). This reduction, together with the fact that \smis remains APX-hard in graphs of maximum degree at most $3$ \cite{AlKa00} (which means in particular that we may assume that $\alpha(G) = \Omega(|E(G)|)$ in the reduction) implies the following:
\begin{theorem}[\cite{Poljak74,Alekseev82} and \cite{AlKa00}]
For any fixed connected graph $H$ which is neither a path nor a subdivision of the claw $K_{1,3}$, \mis is APX-hard in $H$-free graphs.
\end{theorem}

On the positive side, polynomial algorithms are known for $P_6$-free graphs~\cite{GrzesikKPP19} and fork-free graphs~\cite{Alekseev04}.
For paths on at least seven vertices and subdivided claws not contained in the fork, the computational complexity of \smis remains unsettled.

In this work, we start a systematic investigation of the approximability of \smis in $H$-free graphs.
The intuition is that forbidding a fixed graph $H$ as an induced subgraph should imply a drastic change in the structure of independent sets and cliques. 
This idea is at the core of the Erd\H{o}s-Hajnal conjecture: while in random graphs of $\mathcal G(n,1/2)$ the expected maximum of the clique number and the independence number is $O(\log n)$~\cite{erdos47}, this value should be significantly larger for an $H$-free graph.
More formally:

\begin{definition}
A graph $H$ satisfies the \emph{Erd\H{o}s-Hajnal property} if there exists a constant $\delta > 0$ such that every $H$-free graph $G$ with $n$ vertices contains either a clique or an independent set of size $n^{\delta}$.
\end{definition}

\begin{conjecture}[\cite{ErdosHajnal89}]\label{conj:erdoshajnal}
Every graph $H$ satisfies the Erd\H{o}s-Hajnal property.
\end{conjecture}

So far, the Erd\H{o}s-Hajnal conjecture has been verified for only a small number of graphs, namely: all graphs on at most four vertices, the bull\footnote{The bull is the graph obtained by adding a pending vertex to two different vertices of a triangle.}, the cliques, and every graph that can be constructed from them using the so-called \emph{substitution} operation~\cite{Chud14} (we describe this operation in Section~\ref{sec:erdoshajnal}).
Interestingly, for many graphs $H$ satisfying the Erd\H{o}s-Hajnal property, \smis is known to be either polynomial
or at least to admit an $n^{1-\varepsilon}$-approximation algorithm for some $\varepsilon > 0$.
A typical example of this situation is when $H$ is the clique of size $t > 1$.
In that case, Ramsey's theorem can be invoked to get a $n^{\frac{t-2}{t-1}}$-approximation algorithm.
Indeed a $K_t$-free graph always contains an independent set of size at least $n^{\frac{1}{t-1}}$, and the classical proof readily yields a polytime algorithm finding such an independent set.
This leads us to define an approximation weaker version of the Erd\H{o}s-Hajnal property and its companion conjecture:

\begin{definition}
A graph $H$ satisfies the \emph{improved approximation property} if there exists a constant $\varepsilon > 0$ such that \smis admits a (randomized) $n^{1-\varepsilon}$-approximation polynomial algorithm on every $H$-free $n$-vertex graph $G$.
\end{definition}

Here, a \textit{randomized} $\rho$-approximation algorithm is an algorithm which, given an input graph on $n$ vertices, outputs a $\rho$-approximation of the problem \textit{with high probability} (\whp for short), that is with probability at least a function of $n$ tending to $1$ when $n$ goes to infinity.

\begin{conjecture}\label{conj:main}
  Every graph $H$ satisfies the improved approximation property.
\end{conjecture}

We refer to Conjecture~\ref{conj:main} as the \textit{improved approximation conjecture}.
Informally, it states that the inapproximability of \smis in general graphs can be beaten in any proper hereditary class.

\paragraph*{Results and organization of the paper.}

On the one hand, there exist graphs $H$ satisfying the improved approximation property for which the Erd\H{o}s-Hajnal conjecture is still open.
Indeed, as mentioned previously, \smis is polynomial-time solvable in $P_6$-free graphs, whereas it is still open whether $P_5$ satisfies the Erd\H{o}s-Hajnal property.
On the other hand, one may wonder if the satisfiability of the Erd\H{o}s-Hajnal property for a graph $H$ can help designing an approximation algorithm in $H$-free graphs, and more concretely if Conjecture~\ref{conj:erdoshajnal} implies Conjecture~\ref{conj:main}.
In~\cref{sec:erdoshajnal}, we prove that this is almost the case. More precisely, we prove that every graph $H$ satisfying a \emph{constructive} version of the Erd\H{o}s-Hajnal property also satisfies the improved approximation property.
We also show that the improved approximation property is preserved through the substitution operation, which is the one graph operation known to preserve the Erd\H{o}s-Hajnal property.

We then try and obtain better approximation ratios for the improved approximation property: for a given $H$, what is the largest $\varepsilon > 0$ such that \smis admits an $O(n^{1-\varepsilon})$-approximation algorithm in $H$-free graphs? 
We investigate this question in~\cref{sec:someH,sec:negative}.
More precisely, in Section~\ref{sec:someH} we describe some particular properties of graphs $H$ as well as graph operations preserving the improved approximation property in a better way than the substitution.
We also prove that the classical local search algorithm provides a $\sqrt{OPT}$ approximation ratio in $C_4$-free graphs and, more generally, an $O(OPT^{1-1/t})$-approximation algorithm in $K_{t, t}$-free graphs.
Finally, we present in Section~\ref{sec:negative} some negative results concerning the improved approximation property: while \smis can be easily $n^{1/2}$-approximated in triangle-free graphs, we show that this ratio cannot be improved to $n^{1/4-\varepsilon}$ for any $\varepsilon >0$, unless $NP \subseteq BPP$. We also provide a generalization of this result when we forbid all cycles of length $3, \ldots, t$ for a fixed $t \geqslant 3$.

\paragraph*{Notations and definitions.}
For two positive integers $i<j$, we denote the set of integers at least $i$ and at most $j$ by $[i,j]$, while $[i]$ is a short-hand for $[1,i]$. 
All the graphs we consider are simple; they have no multiple edges nor loops.
For a vertex $v$ in a simple graph $G$, $N_G(v)$, or simply $N(v)$ if the graph is unambiguous, denotes the set of neighbors of $v$.
The \emph{closed neighborhood} of $v$ is defined as $N[v] := N(v) \cup \{v\}$.
A \emph{universal vertex} is a vertex whose closed neighborhood is the entire set of vertices.
The size of a maximum independent set of $G$ is denoted by $\alpha(G)$.
The \textit{girth} (resp. \textit{odd girth}) of a graph is the smallest size of an induced cycle (resp. odd cycle) in the graph.
$K_s$, $P_s$, $C_s$ respectively denotes the clique, the path, and the cycle on $s$ vertices, and $K_{s,t}$ is the biclique with $s$ vertices on one side and $t$ on the other side.
The graph $K_3=C_3$ is also called the \emph{triangle}.
The \emph{claw} is the biclique $K_{1,3}$. 
The \emph{fork} is the 5-vertex graph obtained by subdividing one edge of the claw.
For a triple of integers $0 \leqslant i \leqslant j \leqslant k$, the graph $S_{i,j,k}$ is obtained by subdividing one edge of a claw $i-1$ times, a second edge, $j-1$ times, and a third edge $k-1$ times (with the convention that subdividing $-1$ times means removing the edge and its degree-one endpoint).
Observe that with that definition, the family $\{S_{i,j,k}\}_{0 \leqslant i \leqslant j \leqslant k}$ contains the paths.

\section{Constructive Erd\H{o}s-Hajnal and the substitution operation}\label{sec:erdoshajnal}
A graph $H$ is said to satisfy the \emph{constructive Erd\H{o}s-Hajnal} property if there is a constant $\delta > 0$ and a polynomial-time algorithm which takes as input an $H$-free graph $G$, and outputs a clique or an independent set of size at least $|V(G)|^{\delta}$.
We prove that the constructive Erd\H{o}s-Hajnal conjecture implies Conjecture~\ref{conj:main}.
To our knowledge, all the graphs $H$ shown to satisfy the Erd\H{o}s-Hajnal property so far, also satisfy its constructive version.

\begin{theorem}\label{thm:erdoshajnal-implies-approx}
Let $H$ be a graph which satisfies the constructive Erd\H{o}s-Hajnal property with constant\footnote{Notice that \smis in $H$-free graphs is trivial if $H$ has at most $2$ vertices, whereas any graph with at least three vertices cannot satisfy the Erd\H{o}s-Hajnal property with a constant $\delta > 1/2$. This is the reason why we assume $0 < \delta \leqslant 1/2$.} $0 < \delta \leqslant 1/2$. Then $H$ satisfies the improved approximation property with constant $\delta - \delta^2 - \varepsilon$ for any fixed $\varepsilon > 0$. 
\end{theorem}
\begin{proof}
Let $G$ be an $H$-free graph with $n := |V(G)|$. We assume $n \geqslant 2^{\frac{1}{1-(\delta - \delta^2)}}$, since otherwise the problem can be solved optimally in constant time.
We prove that the algorithm described in Figure~\ref{algo:erdoshajnal} provides a $n^{1-(\delta - \delta^2)}$-approximation. In this algorithm, $\texttt{Constructive-}$ $\texttt{Erd\H{o}s-Hajnal}(J)$ represents the polynomial-time algorithm which takes a graph $J$ and outputs a set of at least $|V(J)|^{\delta}$ vertices of $J$ which is either an independent set or a clique.

\begin{figure}[h!]
\begin{algorithmic}[1]
\Require{a graph $G$}
\Ensure{an independent set of $G$}
\State $V' \gets V(G)$
\While{$|V'| \geqslant n^{1-\delta}$}
	\State $X \gets \texttt{Constructive-Erd\H{o}s-Hajnal}(G[V'])$
	\If {$X$ is an independent set of $G$}
    	\State \Return $X$ \label{line:returnset}
	\Else
        \State $V' \gets V' \setminus X$
    \EndIf
\EndWhile
\State \Return $\{v\}$, for an arbitrary chosen $v \in V(G)$ \label{line:returnsingleton}
\end{algorithmic}
\caption{Approximation algorithm for \smis in $H$-free graphs satisfying the constructive Erd\H{o}s-Hajnal property.}
\label{algo:erdoshajnal}
\end{figure}

Let $X$ be the independent set returned by the algorithm. If $X$ is returned through line~\ref{line:returnset}, then by the definition of the $\texttt{Constructive-Erd\H{o}s-Hajnal}$ algorithm, we have $|X| \geqslant n^{(1-\delta)\delta}$ which is obviously an $n^{1-(\delta - \delta^2)}$-approximate solution, since any optimal solution has size at most $n$.

Otherwise, $X$ is returned through line~\ref{line:returnsingleton} and is thus of size $1$. However, in this case, observe that $V(G)$ is partitioned into cliques $C_1$, $\dots$, $C_q$, and the last set $V'$.
Observe that $|V'| < n^{1-\delta}$, and that $|C_i| \geqslant n^{\delta - \delta^2}$ for every $i \in \{1, \dots, q\}$. We thus have $q \leqslant n^{1-(\delta - \delta^2)}$.
But also observe that in that case:
\begin{eqnarray*}
\alpha(G) 	& \leqslant & q + |V'| \\
			& \leqslant & n^{1-(\delta - \delta^2)} + n^{1- \delta} \\
			& \leqslant & 2n^{1-(\delta - \delta^2)} ~~\text{since $n \geqslant 2^{\frac{1}{1-(\delta - \delta^2)}}$} \\
			& \leqslant & n^{1-(\delta - \delta^2) + \varepsilon} ~~\text{as we may assume $n \geqslant 2^{1/\varepsilon}$ for any fixed $\varepsilon > 0$.}
\end{eqnarray*}
\end{proof}

It is natural to ask which kind of graph operations preserves the satisfiability of the  improved approximation property. Given the previous result, natural candidates are graph operations preserving the Erd\H{o}s-Hajnal property. In the following we prove that this is indeed the case concerning the substitution operation.

\begin{definition}
	Let $H_1, H_2$ be two vertex-disjoint graphs and $v_0 \in V(H_1)$. We say that a graph $H$ is \emph{obtained from $H_1$ by substituting $H_2$ at $v_0$} if:
	\begin{itemize}
		\item $V(H)=(V(H_1) \setminus \{v_0\})  \cup V(H_2)$
		\item For $v,v' \in V(H_1) \setminus \{v_0\}$, $vv'$ is an edge in $H$ if and only if it is an edge in $H_1$.
		\item For $v, v' \in V(H_2)$, $vv'$ is an edge in $H$ if and only if it is an edge in $H_2$.
		\item For $v \in V(H_1) \setminus \{v_0\}$, $v' \in V(H_2)$, $vv'$ is an edge in $H$ if and only if $vv_0$ is an edge in $H_1$.
	\end{itemize}
More generally, we say that a graph $H$ is \emph{obtained from $H_1$ and $H_2$ by substitution} if there exists $v_0 \in V(H_1)$ such that $H$ is obtained from $H_1$ by substituting $H_2$ at $v_0$.
\end{definition}

\begin{theorem}\label{thm:substitution}
    Let $H_1$, $H_2$ be two fixed graphs satisfying the improved approximation property.
    Then every graph $H$ obtained from $H_1$ and $H_2$ by substitution satisfies the improved approximation property.
\end{theorem}
Let us start by sketching the idea of our algorithm. We first check whether the number of copies of $H_1$ in $G$ is small. If so, then a randomly chosen subset of vertices of appropriate size will be $H_1$-free \whp, and we will be able to run our approximation algorithm for $H_1$-free graphs. If the number of copies of $H_1$ is large, then we claim that we can find a large subset of vertices inducing an $H_2$-free graph, and we thus run our approximation algorithm for $H_2$-free graphs. 
Each time we run one of our approximation algorithms in an induced subgraph $G[X]$ which is either $H_1$-free or $H_2$-free, either it outputs a solution of size at least $n^{\delta}$ for some constant $\delta$, in which case we are done, or it means that $\alpha(G[X])$ is small, in which case we keep $X$ apart and continue the algorithm on $G[V \setminus X]$ as long as enough vertices survive. If too many vertices were kept apart along the process, it means that $\alpha(G)$ was very small at the beginning, so that any singleton $\{v\}$ is actually an approximated solution. We now prove formally the result.

\begin{proof}
Let $approx_{H_1}(G)$ (resp. $approx_{H_2}(G)$) be a polynomial-time algorithm which takes as input an $H_1$-free graph (resp. $H_2$-free graph) $G$ on $n$ vertices and outputs an $n^{1-\varepsilon_1}$ (resp. $n^{1-\varepsilon_2}$)-approximated solution for the \smis problem in $G$, for some $\varepsilon_1 > 0$ (resp. $\varepsilon_2 > 0$). For the sake of readability, we set $\varepsilon = min \{\varepsilon_1, \varepsilon_2, 0.99\}$, so that $approx_{H_1}$ and $approx_{H_2}$ are $n^{1-\varepsilon}$-approximation algorithms in $H_1$-free graphs and $H_2$-free graphs, respectively\footnote{Our result also holds if $approx_{H_1}$ and $approx_{H_2}$ are exact algorithms (hence with $\varepsilon_1 = \varepsilon_2 = 1$), but, for technical reasons, we view them as $n^{0.01}$-approximation algorithms.}.

 Let $H$ be the graph obtained by substituting $H_2$ at some vertex $v_0 \in V(H_1)$, and let us consider an $H$-free graph $G$. We denote by $n$, $n_1$ and $n_2$ the number of vertices of $G$, $H_1$ and $H_2$, respectively.
 We say that $X \subseteq V(G)$ is a set of \emph{$H_1$-candidates} if there exists a set $K \subseteq V(G)$ of $n_1 - 1$ vertices such that $G[K]$ is isomorphic to $H_1 - \{v_0\}$ and, for every $x \in X$, $G[K \cup \{x\}]$ is isomorphic to $H_1$. Since $G$ is $H$-free, $G[X]$ is $H_2$-free.

 Let  $\gamma = \frac{\varepsilon}{2n_1}$, $\eta = \min(1-\varepsilon, \gamma)$, and $\delta = \frac{\varepsilon \eta} {2+\varepsilon \eta}$.
We prove that the algorithm described in Figure~\ref{algo:substitution} is an $O(n^{1-\delta})$-approximation algorithm for \smis in $H$-free graphs.

\begin{figure}[h!]
		\begin{algorithmic}[1]
			\Require{an $H$-free graph $G$ with $n$ vertices}
			\Ensure{an independent set of $G$}
			\State $i=1$, $V_1 \gets V(G)$
			\While{$|V_i| \geqslant n^{1-\delta}$}
				\If{$G[V_i]$ contains less than $|V_i|^{n_1 - \varepsilon}$ copies of $H_1$} \label{line:copiesH1}
				\State{pick a set $X_i \subseteq V_i$ of size $\lceil|V_i|^{\gamma}\rceil$ uniformly at random} \label{line:findP}
					\If{$G[X_i]$ is $H_1$-free} \Comment{This condition is true \whp} \label{line:h1free}
						\State $W_i \gets approx_{H_1}(G[X_i])$
						\If{ $|W_i| \geqslant n^{\delta}$} 
							 \Return $W_i$ \label{line:callH1} 
						\EndIf
					\Else
						  ~\Return FAIL	\label{line:fail}			
					\EndIf
				\Else
					\State find a set of $H_1$-candidates $X_i \subseteq V_i$ with $|X_i| \geqslant |V_i|^{1-\varepsilon}$ \label{line:H2candidates}
					\State $W_i \gets approx_{H_2}(G[X_i])$
					\Comment $G[X_i]$ is $H_2$-free
					\If{ $|W_i| \geqslant n^{\delta}$ }
						 \Return $W_i$ \label{line:callH2}		
					\EndIf
				\EndIf
				\State $V_{i+1} \gets V_i \setminus X_i$
				\State $i \gets i+1$
			\EndWhile 
			\State \Return $\{v\}$ for an arbitrary vertex $v \in V$ \label{line:singleton}
		\end{algorithmic}
		\caption{Approximation algorithm for \smis in $H$-free graphs, where $H$ is the substitution of $H_1$ and $H_2$.}
		\label{algo:substitution}
	\end{figure}

\begin{lemma}
Algorithm~\ref{algo:substitution} runs in polynomial time.
\end{lemma}
\begin{proof}
An important remark is that at every step $i$, the graph $G[V_i]$ is an induced subgraph of $G$, hence is $H$-free.
In line~\ref{line:copiesH1} (resp.~\ref{line:h1free}), the algorithm runs through all subsets of $n_1$ vertices of $V_i$ (resp. $X_i$), which can be done in $O(n^{n_1})$ time.

Finally, the existence of a set of $H_1$-candidates in line~\ref{line:H2candidates} is ensured by the fact that in that case, $G[V_i]$ contains at least $|V_i|^{n_1 - \varepsilon}$ copies of $H_1$. Hence, by the pigeonhole principle, there must exist $n_1-1$ vertices $V_H \subseteq V_i$ such that $G[V_H]$ induces $H_1 \setminus v_0$ together with a set $X_i \subseteq V_i \setminus V_H$ of size at least $|V_i|^{1-\varepsilon}$ such that for every $x \in X_i$, $G[V_H \cup \{x\}]$ induces $H_1$. Finding the set $V_H$ can be done in $O(|V_i|^{n_1-1})$ time, while finding the set $X_i$ can be done in $O(|V_i|)$ time, since it is sufficient to find the vertices in $V_i \setminus V_H$ with the right neighborhood with respect to $V_H$.
By the definition of $H_1$-candidates, $G[X_i]$ is $H_2$-free, which allows to run $approx_{H_2}$ on $G[X]$ in the next line of the algorithm.
\end{proof}

We now prove that, \whp, the solution $S$ returned by our algorithm is an $O(n^{1 - \delta})$-approximation. To this end, we first prove that \whp it does not return FAIL.

\begin{lemma}\label{lem:H1free}
If the number of copies of $H_1$ in a graph $G$ on $n$ vertices is less than $n^{n_1 - \varepsilon}$, then any subset of vertices of size $\lceil n^{\gamma}\rceil$ picked uniformly at random induces an $H_1$-free graph, with high probability.
\end{lemma}
\begin{proof}
Let $n$ be the number of vertices of $G$, and $P$ be a subset of $\lceil n^{\gamma}\rceil$ vertices picked uniformly at random. For any set $V_H \subseteq V$ inducing $H_1$, the probability that $V_H$ is contained in $P$ is $\frac{{{n-n_1} \choose {|P|-n_1}}}{{n \choose |P|}} < \left(\frac{|P|}{n}\right)^{n_1}$. Hence the probability that $P$ is $H_1$-free is at least $$\left(1 -  \left(\frac{|P|}{n}\right)^{n_1} \right)^{n^{n_1-\varepsilon}} = \left( 1 - \frac{1}{n^{n_1 - \frac{\varepsilon}{2}}}\right)^{n^{n_1 - \varepsilon}}$$
 which tends to $1$ when $n \rightarrow +\infty$.
\end{proof}

Next, if it returns a solution through lines~\ref{line:callH1} or \ref{line:callH2}, then this solution is an independent set of size at least $n^{\delta}$, by definition.
We now deal with the case in which it returns a singleton, through line~\ref{line:singleton}. The aim is to prove that $\alpha(G)$ is at most $O(n^{1-\delta})$.
Let $q+1$ be the largest value of $i$ in the execution of the algorithm (i.e., $|V_{q+1}| < n^{1-\delta}$). 
The vertex-set $V$ is thus partitioned into $X_1$, $\dots$, $X_q$, and $V_{q+1}$.
Hence we have $\alpha(G)  \leqslant  |V_{q+1}| + \sum_{i=1}^q \alpha(G[X_i])$. Since $|V_{q+1}| < n^{1 - \delta}$, we only need to upper bound the second part.
\begin{lemma}
With the above definitions, $\sum_{i=1}^q \alpha(G[X_i]) \leqslant n^{1-\delta}$.
\end{lemma}
\begin{proof}
 Recall that we have $X_i \subseteq V_i$, where $V_i = V \setminus \bigcup_{j=1}^{i-1} X_i$, and, for each $i \in [q]$, we have constructed an independent set $W_i \subseteq X_i$. All these sets have the following properties:

\begin{enumerate}
	\item $|V_{q+1}| < n^{1-\delta}$, by definition of $q$.
	\item $|X_i| \geqslant |V_i|^{\eta} \geqslant n^{\eta(1-\delta)}$. Indeed, if $X_i$ is defined in line~\ref{line:findP}, then it is of size  at least $|V_i|^{\gamma}$, whereas if it is defined in line~\ref{line:H2candidates}, it is of size at least $|V_i|^{1-\varepsilon}$, and $\eta = \min(1-\varepsilon, \gamma)$.\label{enum:lbXi}
	\item $|W_i| < n^{\delta}$, otherwise we would have returned it.
	\item $\alpha(G[X_i]) \leqslant |W_i| \cdot |X_i|^{1-\varepsilon}$, since $W_i$ is returned by $approx_{H_1}$ or $approx_{H_2}$, which are approximation algorithms applied to $G[X_i]$.
\end{enumerate}

Now, we have the following:

\begin{eqnarray*}
	\sum_{i=1}^q \alpha(G[X_i]) 	& \leqslant &  \sum_{i=1}^q |W_i| \cdot |X_i|^{1-\varepsilon} \\
				& \leqslant & n^{\delta} \sum_{i=1}^q |X_i|^{1-\varepsilon}
\end{eqnarray*}

We then need the following technical lemma.
	\begin{lemma}\label{lem:calcul}
		Let $(a_i)_{i=1...q}$ be some positive numbers (with $q \in \mathbb{N}$) such that $\sum_{i=1}^q a_i = N$ and $a_i \geqslant k > 0$ for all $i \in \{1, \dots, q\}$. Then $\sum_{i=1}^q a_i^{\zeta} \leqslant N k^{\zeta - 1}$ for any $0 < \zeta < 1$.
	\end{lemma}
\begin{proof}
We have:
$$\sum_{i=1}^q a_i^{\zeta} \leqslant \left(\frac{N}{q}\right)^{\zeta} \cdot q = N \cdot \left(\frac{N}{q}\right)^{\zeta-1} \leqslant N k^{\zeta-1}.$$
\end{proof}

Using the above lemma together with item~\ref{enum:lbXi} of the previous properties in order to lower bound each $|X_i|$, we obtain:
	
	\begin{eqnarray*}
	\sum_{i=1}^q \alpha(G[X_i]) 	& \leqslant & n^{\delta} \left( \sum_{i=1}^q |X_i| \right) n^{\eta(1-\delta)(1-\varepsilon-1)} \\
		& \leqslant & n^{1 -\varepsilon \eta(1-\delta) + \delta} ~~~\text{since $\sum_{i=1}^q |X_i| \leqslant n$} \\
		& \leqslant & n^{1 - \delta}~~~\text{because $\delta = \frac{\varepsilon \eta}{2+\varepsilon \eta}$, hence $\varepsilon \eta(1-\delta) = 2\delta$}
	\end{eqnarray*}
   \end{proof}    
   	Hence, any solution of size $1$ is an $O(n^{1-\delta})$-approximation in this case, which concludes the proof.
\end{proof}

\section{Better approximation ratios}\label{sec:someH}

In this section we improve over the ratio given by \cref{thm:substitution} for some graphs $H$ that can be built by a sequence of substitutions from graphs $H'$ such that \mis is polynomial-time solvable in $H'$-free graphs.
Furthermore, we present deterministic algorithms.

\subsection{Adding a universal vertex}

Let $H^{+u}$ be the graph $H$ augmented by a universal vertex, i.e., we add one vertex adjacent to all the vertices of $H$.

\begin{lemma}\label{lem:HtoHplusOPT}
  Let $0 \leqslant \gamma < 1$ be a real number and $H$ be a graph such that \mis admits an $\text{OPT}^\gamma$-approximation $\mathcal A$ in $H$-free graphs.
  Then it also admits an $\text{OPT}^{\frac{1}{2-\gamma}}$-approximation $\mathcal A^{+u}$ in $H^{+u}$-free graphs. 
\end{lemma}
\begin{proof}
  Let $G$ be the input graph, thus $\text{OPT} := \alpha(G)$.
  The base case of the algorithm is when $G$ does not contain any vertex, and we correctly report the empty set as optimum solution.
  Otherwise $G$ has at least one vertex, say $v_1$.
  We run the approximation $\mathcal A$ on $G[N(v_1)]$.
  $G$ being $H^{+u}$-free, the subgraph induced by the open neighborhood of any vertex is indeed $H$-free.
  Let $S_1$ be the returned solution.
  By assumption, $|S_1| \geqslant \alpha(G[N(v_1)])^{1-\gamma}$.
  For what follows, the knowledge of the value $\text{OPT}$ would help.
  Unfortunately we will make some recursive calls to $\mathcal A^{+u}$, so exhaustively guessing this value would result in an exponential running time.
  Instead we will branch but the branching tree will only have at most $n := \lvert V(G) \rvert$ leaves.
  More precisely the tree will be a so-called \emph{comb}, i.e., a path where all the vertices except one end has an additional private neighbor.
  We eventually output the best solution found among all the leaves.

  We inductively run $\mathcal A^{+u}$ on $G - N[v_1]$, which produces a tree $T$ with at most $n-1$ leaves.
  And we output the best solution among $S$ and all the solutions at the leaves of $T$ augmented by the vertex $v_1$.
  This algorithm returns an independent set since $v_1$ is by definition non-adjacent to any vertex of $G - N[v_1]$.
  The running time of our algorithm satisfies $f_{\mathcal A^{+u}}(n) = f_{\mathcal A}(n_1-1) + f_{\mathcal A^{+u}}(n-n_1) + O(1)$ (with $n_1=|N(v_1)|$).
  Hence $f_{\mathcal A^{+u}}(n) = O(\max \{f_{\mathcal A}(n),n\})$, and in the likely event that $\mathcal A$ is \emph{not} sublinear, $\mathcal A^{+u}$ has the same running time as $\mathcal A$ up to a multiplicative constant factor.

  We shall now show that $\mathcal A^{+u}$ is indeed a $\text{OPT}^{\frac{1}{2-\gamma}}$-approximation.
  We denote by $v_1, \ldots, v_p$ with $p \leqslant n$, the vertices added along the path to the deepest leaf of $T$.
  We denote by $S_1, \ldots, S_p$ the sets returned by $\mathcal A$ such that $S_i$ is computed in the graph $G'_i := G[N(v_i) \setminus \bigcup_{j<i} N[v_j]]$.
  We also define $G_i := G[N[v_i] \setminus \bigcup_{j<i} N[v_j]]$, and $R_i := G - \bigcup_{j<i} N[v_j]$.
  Observe that $\{V(G_1), \ldots, V(G_p)\}$ is a partition of $V(G)$, as well as, $\{V(G_1), \ldots, V(G_i), V(R_{i+1})\}$ for every $i \in [p-1]$.

  Let, if it exists, $S_h$ be the first solution returned by $\mathcal A$ when called on $G'_h$ such that $\lvert S_h \rvert \geqslant \alpha(R_h)^{1-\frac{1}{2-\gamma}}$.
  We claim that the solution output at this leaf, namely $S'_h := S_h \cup \{v_1, \ldots, v_{h-1}\}$ is an $\text{OPT}^{\frac{1}{2-\gamma}}$-approximation.
  If such an $S_h$ does not exist, we show the same statement where $S_h = \emptyset$ and $h-1 = p$.
  
  We upperbound $\alpha(G_i)$ for every $i \in [h-1]$.
  By definition of $S_h$, it holds that $\lvert S_i \rvert < \alpha(R_i)^{1-\frac{1}{2-\gamma}}$ for any $i \in [h-1]$.
  Due to the approximation ratio of $\mathcal A$, it holds that:
  $$ \alpha(G_i)^{1-\gamma} \leqslant |S_i| < \alpha(R_i)^{1-\frac{1}{2-\gamma}} = \alpha(R_i)^{\frac{1-\gamma}{2-\gamma}},$$
  hence $\alpha(G_i) < \alpha(R_i)^{\frac{1}{2-\gamma}} \leqslant \alpha(G)^{\frac{1}{2-\gamma}}$.
  Thus, $$\text{OPT} = \alpha(G) \leqslant \alpha(R_h) + \sum\limits_{i \in [h-1]} \alpha(G_i) \leqslant \lvert S_h \rvert \alpha(R_h)^{\frac{1}{2-\gamma}} + \sum\limits_{i \in [h-1]} \alpha(G)^{\frac{1}{2-\gamma}}$$ $$ \leqslant \lvert S_h \rvert \alpha(G)^{\frac{1}{2-\gamma}} + (h-1) \alpha(G)^{\frac{1}{2-\gamma}} = (\lvert S_h \rvert + h - 1)\alpha(G)^{\frac{1}{2-\gamma}} = \lvert S'_h \rvert \alpha(G)^{\frac{1}{2-\gamma}}=\lvert S'_h \rvert \text{OPT}^{\frac{1}{2-\gamma}}.$$
  Therefore $\mathcal A^{+u}$ is an $\text{OPT}^{\frac{1}{2-\gamma}}$-approximation for \smis in $H^{+u}$-free graphs.
\end{proof}

\subsection{Locally easy graphs}

We say that a graph $H$ is \emph{\nice} if it has a universal vertex $v$ such that there is a polynomial-time algorithm for \smis in $H-\{v\}$-free graphs.   
Up to now, the three maximal graphs $H$ for which we know that \smis is polynomial-time solvable on $H$-free graphs are $P_6$ \cite{Grzesik17}, the fork \cite{Alekseev04,Lozin08}, and $tK_{1,3}$ (or $t$claw) \cite{BrandstadtM18a} (see Figure~\ref{fig:maximal-locally-easy} for the corresponding maximal \nice graphs).

\begin{figure}
  \begin{tikzpicture}
    \def\s{0.6}
    \foreach \i in {1,...,6}{
      \node[draw, circle] (a\i) at (\s * \i, 0) {};
    }
    \draw (a1) -- (a2) -- (a3) -- (a4) -- (a5) -- (a6) ;
    \node[draw, circle] (a) at (\s * 3.5, 1) {}; 
    \foreach \i in {1,...,6}{
      \draw (a) -- (a\i) ;
    }

    \begin{scope}[xshift=6.2cm]
      \node[draw, circle] (b1) at (0, \s) {};
      \node[draw, circle] (b2) at (0, 0) {};
      \node[draw, circle] (b3) at (0, -\s) {};
      \node[draw, circle] (b4) at (-0.5 * \s, 2 * \s) {};
      \node[draw, circle] (b5) at (0.5 * \s, 2 * \s) {};
      \draw (b1) -- (b2) -- (b3) ;
      \draw (b4) -- (b1) -- (b5) ;

      \node[draw,circle] (b) at (-2 * \s, 3 * \s) {} ;
    \foreach \i in {1,...,5}{
      \draw (b) -- (b\i) ;
    }
    \end{scope}

    \begin{scope}[xshift=7cm]
      \foreach \i in {1,2,4,5}{
        \node[draw, circle] (c\i1) at (2.2 * \s * \i, 0) {};
        \node[draw, circle] (c\i2) at (2.2 * \s * \i - 0.75 * \s, \s) {};
        \node[draw, circle] (c\i3) at (2.2 * \s * \i, \s) {};
        \node[draw, circle] (c\i4) at (2.2 * \s * \i + 0.75 * \s, \s) {};
        \draw (c\i1) -- (c\i2) ;
        \draw (c\i3) -- (c\i1) -- (c\i4) ;
      }
      \node[draw, circle] (c) at (6.6 * \s,2) {} ;
      \node at (6.6 * \s,\s) {$\ldots$} ;
      \foreach \i in {1,2,4,5}{
        \foreach \j in {1,...,4}{
          \draw (c) -- (c\i\j) ;
        }
      }
    \end{scope}
  \end{tikzpicture}
  \caption{The three maximal \nice graphs $H$ constructed from $P_6$, the fork, and $tK_{1,3}$, respectively.}
  \label{fig:maximal-locally-easy}
\end{figure}
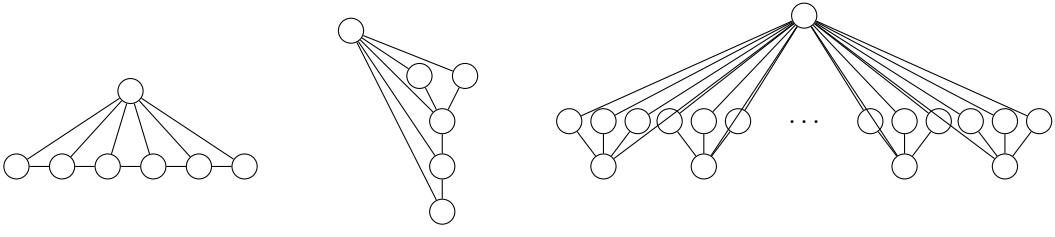

The following is an immediate consequence of \cref{lem:HtoHplusOPT}.
Therein we recall that the constant $\gamma$ may take value 0. 

\begin{theorem}\label{thm:sqrtOPT}
  For any \nice $H$, \mis can be $\sqrt \text{OPT}$-approximated on $H$-free graphs.
\end{theorem}
And in particular, there is a $\sqrt n$-approximation for triangle-free graphs.
\cref{lem:HtoHplusOPT} also yields the following approximation ratio in $K_{t+1}$-free graphs.

\begin{theorem}\label{thm:clique-free}
  For any $t \geqslant 1$, \mis can be $\text{OPT}^{1-\frac{1}{t}}$-approximated on $K_{t+1}$-free graphs.
\end{theorem}
\begin{proof}
  We show this statement by induction.
  The base case says that we can exactly solve in polynomial time \mis in edgeless graphs, which is obviously true.
  We assume that the statement is true for a fixed $t$.
  As $K_{t+2}=(K_{t+1})^{+u}$, \cref{lem:HtoHplusOPT} implies that \mis can be $\text{OPT}^{\frac{1}{2-(1-\frac{1}{t})}}$-approximated on $K_{t+2}$-free graphs.
  Furthermore, ${\frac{1}{2-(1-\frac{1}{t})}} = \frac{1}{1+\frac{1}{t}} = \frac{t}{t+1} = 1 - \frac{1}{t+1}$.
  Therefore we do obtain an $\text{OPT}^{1 - \frac{1}{t+1}}$-approximation on $K_{t+2}$-free graphs.
\end{proof}

We say that a graph $H$ is \emph{\tnice} if it has a set $U$ of $t$ universal vertices such that a polynomial algorithm is known for \smis in $H-U$-free graphs.
Informally, these graphs are obtained by replacing universal vertices of Figure~\ref{fig:maximal-locally-easy} by a $k$-clique.
The previous result readily generalizes from $K_{t+1}$-free to $H$-free graphs with $H$ \pnice{(t+1)}, with the same proof.
\begin{corollary}\label{thm:tpo-loc-easy}
Let $t$ be a non-negative integer. For any \pnice{t} $H$, \mis can be $\text{OPT}^{1-\frac{1}{t+1}}$-approximated on $H$-free graphs.
\end{corollary}

As we will see in \cref{sec:improving} for forbidden graphs $H$ containing a triangle, it is unlikely to improve the approximation ratio below $n^{1/4}$.
However if $H$ is a star, constant-approximations are achievable.
It is known that \smis is polynomial-time solvable on claw-free graphs \cite{Minty80} (i.e., $K_{1,3}$-free graphs) while it is APX-hard on $K_{1,4}$-free graphs (see for instance \cite{Alekseev82}).
The greedy algorithm (or actually any sensible algorithm) gives an $s$-approximation in $K_{1,s}$-free graphs.
The ratio was improved to arbitrarily close to $\frac{s-1}{2}$ by Halld\'orsson.
\begin{theorem}[\cite{Halldorsson95}]
  For every $s \geqslant 4$ and $\varepsilon > 0$, \mis is $\frac{s-1}{2}+\varepsilon$-approximable on $K_{1,s}$-free graphs.
\end{theorem}

\subsection{\textsc{Local Search} for $K_{t,t}$-free graphs}

Here we analyze the performance of the \textsc{$t$-Local Search} algorithm in $K_{t,t}$-free graphs.
Usually for the particular case of the biclique, ``$K_{t,t}$-free'' is intended as ``no $K_{t,t}$ as a subgraph''.
Here we still mean ``no $K_{t,t}$ as an \emph{induced} subgraph'', since our algorithm works even in this more general setting.
For a fixed integer $t > 2$, \textsc{$t$-Local Search} takes as input a graph $G$, and construct an independent set $S$ from a single vertex.
Then, it tries to improve $S$ in the following way: whenever there exist two sets $X \subseteq S$ (note that $X$ can possibly be empty) and $Y \subseteq V \setminus S$ such that $0 \leqslant |X| < |Y| \leqslant t$ and $(S \setminus X) \cup Y$ is an independent set, it replaces $S$ by $(S \setminus X) \cup Y$ (if there are several choices, it chooses an arbitrary one). When $S$ can no longer be improved, it outputs it. Each improvement takes $O(n^{2t})$ time, and the number of such improvements is at most $n$, since the size of $S$ increases by at least one at each step. Hence, the algorithm takes polynomial time. 
In the following theorem, we prove that this simple algorithm provides an $O(OPT^{1-1/t})$-approximation whenever the input graph is $K_{t,t}$-free.
In particular, \textsc{$2$-Local Search} is an $O(\sqrt{OPT})$-approximation in $C_4$-free graphs.
It came to our knowledge that the same result was obtained independently by Dvořák, Feldmann, Rai, and Rzążewski~\cite{privateCommunication}.

\begin{theorem}\label{thm:Ktt-free}
For any fixed $t \geqslant 2$, \textsc{$t$-Local Search} is an $O(OPT^{1-1/t})$-approximation in $K_{t, t}$-free graphs.
\end{theorem}
\begin{proof}
Let $S$ be the solution returned by the algorithm, and $O$ be a fixed optimal solution. The objective is to bound $|O'|$ in terms of $|S'|$, where $O' := O \setminus S$ and $S' := S \setminus O$. To this end, let us consider $\mathcal{B} := G[S' \cup O']$ the bipartite graph induced by $S' \cup O'$. Let $k := |S'|$.
We partition $O'$ into $D^-$ and $D^+$, where $D^-$ are the vertices of $O'$ whose degree within $S'$ is at most $t-1$, and thus $D^+$ are the vertices of $O'$ whose degree within $S'$ is at least $t$. We now bound the sizes of $D^-$ and $D^+$ separately.
\begin{itemize}
	\item Let us partition $D^-$ into classes $D^-_1$, $\dots$, $D^-_q$ with respect to the equivalence relation $u \sim v$ if and only if $N_{S'}(u) = N_{S'}(v)$. By definition of $D^-$ we have $q \leqslant \sum_{i=1}^{t-1}{k \choose i}$. Then, we claim that for every $i \in \{1, \dots, q\}$, we have $|D^-_i| \leqslant t-1$. Indeed, we must have $|D^-_i| \leqslant |N_{S'}(D^-_i)|$, since otherwise the algorithm would have replaced $S$ by $(S \setminus N_{S'}(D^-_i)) \cup D^-_i$. This proves $|D^-| \leqslant (t-1) \sum_{i=1}^{t-1}{k \choose i}$.\\
	
	\item For a set $X \subseteq S'$, let $I_X := \bigcap_{x \in X} N_{O'}(x)$. Observe that if $|X| = t$, then necessarily $I_X \subseteq D^+$, and moreover $|I_X| \leqslant t-1$, since otherwise the graph would have an induced $K_{t, t}$. Finally, we have $D^+ = \bigcup_{X \subseteq S', |X|=t} I_X$, which proves that $|D^+| \leqslant {k \choose t}(t-1)$.	
\end{itemize}
Hence we have $|O'| \leqslant |D^-| + |D^+| \leqslant (t-1) \sum_{i=1}^{t-1}{k \choose i} + {k \choose t}(t-1) = O(k^t)$.
\end{proof}

\section{Graphs without short cycles} \label{sec:negative}

In this section we show that the strong inapproximability of \mis in general graphs survives, albeit in a less severe form, on graphs without small cycles.
More quantitatively, we show that for any positive integer $\gamma$, there is a constant $\beta = \Theta(1/\gamma)$ depending only on $\gamma$, such that an $n^\beta$-approximation of \mis in graphs with girth $\gamma$ is unlikely.

\subsection{Triangle-free graphs}\label{sec:triangle-free}

While~\cref{lem:HtoHplusOPT} implies an $n^{1/2}$-approximation of \mis in triangle-free graphs, a natural question is how much the ratio's exponent can be decreased.
In this section we provide a lower bound for it.

The following result will be made obsolete twice.
Indeed we will then generalize its statement from triangle-free, that is girth 4, to graphs with any constant girth.
Then in \cref{sec:improving} we will present a stronger inapproximability result of $\Omega(n^{1/4-\varepsilon})$.
Nevertheless we choose to keep its proof as it is simpler, easier to follow, and self-contained.
Furthermore, it contains all the ideas necessary to achieve the subsequent results. 

\begin{theorem}\label{thm:triangle-free}
	For any $\varepsilon > 0$, it is NP-hard to distinguish between triangle-free graphs $G$ on $n$ vertices satisfying
	\begin{itemize}
		\item $\alpha(G) \leqslant n^{5/6 - \varepsilon}$, and
		\item $\alpha(G) \geqslant n^{1 - \varepsilon}$.
	\end{itemize}
	So for any $\varepsilon > 0$, \mis cannot be approximated within ratio $n^{1/6 - \varepsilon}$ in triangle-free graphs unless NP~$\subseteq$~BPP.
\end{theorem}

\begin{proof}
	Let $\varepsilon > 0$ be an arbitrarily small real value, and $\varepsilon := 3 \varepsilon$.
	We perform a randomized reduction from an infinite set of graphs $H$ admitting the following gap:
	Positive instances have stable sets of size at least $|V(H)|^{1-\varepsilon}$ whereas negative instances have no stable set of size $|V(H)|^\varepsilon$.
	It is known that distinguishing between these two cases is NP-hard for randomized reductions \cite{Hastad96}, and even for deterministic ones \cite{Zuckerman07}.
        
	\medskip
	
	\textbf{Reduction.}
	Given an $N$-vertex graph $H$, we construct a triangle-free graph $G$ in the following way.
	We transform every vertex $v$ of $H$ into an independent set $I(v)$ of size $s := N^5$.
	For every edge $uv \in E(H)$, we put a random bipartite graph between $I(u)$ and $I(v)$: for each pair of vertices $x \in I(u)$, $y \in I(v)$, we independently add an edge $xy$ to $E(G)$ with probability $p := N^{-4-\eta}$ with $\eta := 2\varepsilon / 3$.
        We denote by $G_\triangle$ the graph thus obtained.
        A key property is that $G_\triangle$ contains only \emph{few} triangles.
	For each triangle in $G_\triangle$, we remove all three vertices of it.
	We call that phase the \emph{triangle removal}, and we denote by $G$ the triangle-free graph that arises when that phase comes to an end.
        We further assume that $N$ is larger than the smallest integral constant $N_0$ for which for every $N \geqslant N_0$, $N^{3\eta} > N^{2.5\eta} + 10 N^{2\eta} \ln N$, $N^{-\varepsilon} > 6N^{-2\eta}$, $2^{17/6}N^{\eta/3} < N^\varepsilon$, and $N^{-\eta} < 10^{-100}$.
        In particular the second and third inequalities hold for sufficiently large $N$ since $\eta < \varepsilon = 3 \varepsilon < 2\eta$.
        The hardness of approximation \cite{Hastad96,Zuckerman07} still holds since instances with less than a constant number of vertices can be solved optimally in constant time.  
	
	\begin{lemma}\label{lem:edgeSimulator}
	  For every edge $uv \in E(H)$, the probability that there exist two sets $A \subset I(u), B \subset I(v)$ both of size $N^{4+2\eta}$ without any edge between $A$ and $B$ is at most $e^{-N^{4+2.5\eta}}$.
          Thus, with high probability, this event does not happen.
	\end{lemma}
	\begin{proof}
	  The probability that there is no edge between two fixed sets $A$ and $B$ of size $N^{4+2\eta}$ is: 
	  $$(1-p)^{|A| \cdot |B|}=(1-\frac{1}{N^{4+\eta}})^{N^{2(4+2\eta)}} \leqslant e^{-N^{4+3\eta}}.$$
          By the union bound, the probability that there is at least one such pair of sets is at most:
          $${N^5 \choose N^{4+2\eta}}^2 e^{-N^{4+3\eta}} \leqslant N^{10 N^{4+2\eta}} e^{-N^{4+3\eta}} = e^{-N^{4+3\eta}+10 N^{4+2\eta} \ln N} \leqslant e^{-N^{4+2.5\eta}}.$$
\end{proof}
	
	\begin{lemma}\label{lem:fewTriangles}
	  The expected number of triangles in $G_\triangle$ is at most $N^{6-3\eta}$.
          Furthermore $|V(G_\triangle)|-|V(G)|$ is at most $3N^{6-2\eta}$ with probability at least $1 - N^{-\eta}$.
	\end{lemma}
	\begin{proof}
	  The expected number of triangles in $G_\triangle$ is:
	  $$\mathbb E(\#(\triangle,G_\triangle)) \leqslant (sN)^3p^3 = (N^6)^3N^{-12-3\eta} = N^{6-3\eta}.$$
          By Markov's inequality, $\mathbb P(\#(\triangle,G_\triangle) \geqslant N^{6-2\eta}) \leqslant N^{6-3\eta}/N^{6-2\eta} = N^{-\eta}$.
	\end{proof}
	
	Let $n$ be the number of vertices of $G$ (after the triangle removal).
        By the previous lemma $n := |V(G)| > N^6/2$, with high probability.
	
	\medskip
	
	\textbf{If $H$ is a YES-instance, there is a stable set of size $\mathbf{n^{1-\varepsilon}}$ in $G$.}
	We assume that $H$ is a YES-instance, so there is a stable set $S$ in $H$ such that $|S| \geqslant N^{1-\varepsilon}$.
	By construction, $S_{G_\triangle} := \bigcup_{u \in S} I(u)$ is a stable set in $G_\triangle$ of size $s|S| \geqslant N^{6-\varepsilon}$.
        By Lemma~\ref{lem:fewTriangles}, $S_{G_\triangle} \cap V(G)$ is an independent set in $G$ of size, w.h.p., at least $N^{6-\varepsilon} - 3N^{6-2\eta} > N^{6-\varepsilon}/2 > n^{1-\varepsilon/6}/2 = n^{1 - \varepsilon/2} /2 > n^{1 - \varepsilon}$.
	
	\medskip
	
	\textbf{If H is a NO-instance, there is no stable set of size $\mathbf{n^{5/6+\varepsilon}}$ in G.} 
        Let $S_G$ be an independent set of $G$ and let $S := \{v \in V(H)~\text{such that}~|I(v) \cap S_G| \geqslant N^{4+2\eta}\}$.
        If $H$ is a NO-instance, then there is no stable set in $H$ of size more than $N^\varepsilon$.
        By a union bound of applications of Lemma~\ref{lem:edgeSimulator} to all pairs of vertices of $S$, w.h.p $S$ is an independent set of $H$, which implies that $|S|<N^{\varepsilon}$.
        Thus $|S_G| < sN^{\varepsilon}+N^{4+2\eta}(N-N^{\varepsilon}) < N^{5+\varepsilon}+N^{5+2\eta} < 2N^{5+2\eta} = 2N^{5(1+2\eta/5)} < 2^{11/6}n^{5/6 + \eta / 3} < n^{5/6+\varepsilon}$.
\end{proof}

\subsection{Graphs with higher girth}\label{sec:high-girth}

Monien, Speckenmeyer and Murphy independently found improved approximations when the girth, actually even the odd girth, is any constant $\gamma$.

\begin{theorem}[\cite{Monien85,Murphy92}]\label{thm:MonienMurphy}
  \mis admits a polynomial-time $n^{\frac{2}{\gamma - 1}}$-approximation on graphs with odd girth $\gamma$.
\end{theorem}

In particular, the result implies an $n^{1/2}$-approximation for triangle-free graphs, an $n^{1/3}$-approximation for $\{C_3,C_5\}$-free graphs, an $n^{1/4}$-approximation for $\{C_3,C_5,C_7\}$-free graphs, etc.
On the complexity side, the construction of Theorem~\ref{thm:triangle-free} where the probability $p$ of having an edge between $I(u)$ and $I(v)$ with $uv \in E(H)$ is now set to $N^{-2(\gamma-1)-\eta}$ and the size $s$ of each $I(u)$ is set to $N^{2\gamma-1}$ yields a polynomial gap on $C_{\gamma}$-free graphs, and even on graphs with girth $\gamma+1$.

\begin{theorem}\label{thm:high-girth}
	For any $\varepsilon > 0$, it is NP-hard to distinguish between graphs $G$ with $n$ vertices and girth $\gamma+1$ satisfying
	\begin{itemize}
		\item $\alpha(G) \leqslant n^{\frac{2\gamma-1}{2\gamma} - \varepsilon}$, and
		\item $\alpha(G) \geqslant n^{1 - \varepsilon}$.
	\end{itemize}
	Hence, for any $\varepsilon > 0$, \mis cannot be approximated within ratio $n^{\frac{1}{2\gamma} - \varepsilon}$ in graphs with girth $\gamma + 1$ unless NP $\subseteq$ BPP.
\end{theorem}

\begin{proof}
  We do the same reduction as in Theorem~\ref{thm:triangle-free} with the following modifications.
  We now set $\varepsilon := \gamma \varepsilon$, $s := N^{2\gamma-1}$, $p := N^{2(\gamma-1)-\eta}$, and $\eta := \frac{\gamma-1}{\gamma} \varepsilon$.
  We denote by $G_\circ$ the graph obtained before the removal step.
  For every cycle of length at most $\gamma$, we remove all the vertices of the cycle from the graph.
  When this short cycle removal ends, the graph has girth at least $\gamma + 1$.
  We call $G$ the obtained graph.
  
\begin{lemma}\label{lem:edgeSimulatorGen}
	  For every edge $uv \in E(H)$, the probability that there exist two sets $A \subset I(u), B \subset I(v)$ both of size $N^{2(\gamma-1)+2\eta}$ without any edge between $A$ and $B$ is at most $e^{-N^{2(\gamma-1)+2.5\eta}}$.
          Thus, with high probability, this event does not happen.
\end{lemma}
	
	\begin{proof}
	  The probability that there is no edge between two fixed sets $A$ and $B$ of size $N^{2(\gamma-1)+2\eta}$ is: 
	  $$(1-p)^{|A| \cdot |B|}=(1-\frac{1}{N^{2(\gamma-1)+\eta}})^{N^{2(2(\gamma-1)+2\eta)}} \leqslant e^{-N^{2(\gamma-1)+3\eta}}.$$
          By the union bound, the probability that there is at least one such pair of sets is at most:
          $${N^{2 \gamma - 1} \choose N^{2(\gamma-1)+2\eta}}^2 e^{-N^{2(\gamma-1)+3\eta}} \leqslant N^{10 N^{2(\gamma-1)+2\eta}} e^{-N^{2(\gamma-1)+3\eta}}$$
          $$= e^{-N^{2(\gamma-1)+3\eta}+10 N^{2(\gamma-1)+2\eta} \ln N} \leqslant e^{-N^{2(\gamma-1)+2.5\eta}}.$$
	\end{proof}
	
	\begin{lemma}\label{lem:fewShortCycles}
	  The expected number of cycles of length at most $\gamma$ in $G_\circ$ is at most $N^{(2-\eta)\gamma}$.
          Furthermore $|V(G_\circ)|-|V(G)|$ is at most $\gamma N^{2\gamma-\eta(\gamma-1)}$ with probability at least $1 - N^{-\eta}$.
	\end{lemma}
	
	\begin{proof}
	  The expected number of cycles of length at most $\gamma$ in $G_\circ$ is:
	  $$\mathbb E(\#(C_{3 \rightarrow \gamma},G_\circ)) \leqslant \gamma(sN)^\gamma p^\gamma = \gamma N^{2 \gamma^2} N^{(-2(\gamma-1)-\eta)\gamma} = \gamma N^{(2-\eta)\gamma}.$$
          By Markov's inequality, $\mathbb P(\#(C_{3 \rightarrow \gamma},G_\circ) \geqslant \gamma N^{2\gamma-\eta(\gamma-1)}) \leqslant N^{(2-\eta)\gamma}/N^{2\gamma-\eta(\gamma-1)} = N^{-\eta}$.
	\end{proof}
	
	Let $n$ be the number of vertices of $G$ (after the short cycle removal).
        By the previous lemma $n := |V(G)| > N^{2 \gamma}/2$, with high probability.
	
	\medskip
	
	\textbf{If $H$ is a YES-instance, there is a stable set of size $\mathbf{n^{1-\varepsilon}}$ in $G$.}
	We assume that $H$ is a YES-instance, so there is a stable set $S$ in $H$ such that $|S| \geqslant N^{1-\varepsilon}$.
	By construction, $S_{G_\circ} := \bigcup_{u \in S} I(u)$ is a stable set in $G_\circ$ of size $s|S| \geqslant N^{2\gamma-\varepsilon}$.
        By~\cref{lem:fewShortCycles}, $S_{G_\circ} \cap V(G)$ is an independent set in $G$ of size, w.h.p., at least $N^{2\gamma-\varepsilon} - \gamma N^{2\gamma-\eta(\gamma-1)} > N^{2\gamma-\varepsilon}/2 > n^{1-\varepsilon/{2\gamma})}/2 = n^{1 - \varepsilon/2} /2 > n^{1 - \varepsilon}$.
	
	\medskip
	
	\textbf{If H is a NO-instance, there is no stable set of size $\mathbf{n^{(2\gamma-1)/(2\gamma)+\varepsilon}}$ in G.} 
        Let $S_G$ be an independent set of $G$ and let $S := \{v \in V(H)~\text{such that}~|I(v) \cap S_G| \geqslant N^{2(\gamma-1)+2\eta}\}$.
        If $H$ is a NO-instance, then there is no stable set in $H$ of size more than $N^\varepsilon$.
        By a union bound of applications of~\cref{lem:edgeSimulatorGen} to all pairs of vertices of $S$, w.h.p $S$ is an independent set of $H$, which implies that $|S|<N^{\varepsilon}$.
        Thus $|S_G| < sN^{\varepsilon}+N^{2(\gamma-1)+2\eta}(N-N^{\varepsilon}) < N^{2\gamma-1+\varepsilon}+N^{2\gamma-1+2\eta} < 2N^{2\gamma-1+2\eta} = 2N^{(2 \gamma -1)(1+2\eta/(2 \gamma -1))} < 2^{(4\gamma-1)/(2\gamma)}n^{(2\gamma-1)/(2\gamma) + \eta / \gamma} < n^{(2\gamma-1)/(2\gamma)+\varepsilon}$.
\end{proof}

Let us note that there is still a 4-fold multiplicative factor in the exponent between the approximation ratios of Theorem~\ref{thm:MonienMurphy}, namely $n^{2/\gamma}$, and the hardness ratios of $n^{1/(2\gamma) - o(1)}$ in Theorem~\ref{thm:high-girth}.
It is an interesting open question to bridge this gap.

An even hole is an induced cycle of even length at least 4.
Even-hole-free graphs are $\{C_4,C_6,C_8,\ldots \}$-free graphs.
The computational complexity of \smis on even-hole-free graphs is still unknown.
An FPT algorithm was established recently \cite{Husic19}.
We observe that Local Search readily gives a PTAS for that problem.
We leave the existence of an EPTAS as an open problem.
\begin{observation}\label{thm:even-hole-free}
  \mis can be $(1+\varepsilon)$-approximated in time $n^{O(1/\varepsilon)}$ on even-hole-free graphs.
\end{observation}
\begin{proof}
  The graph induced by the symmetric difference between any two feasible solutions is bipartite and even-hole-free, hence it is a forest.
  Let $S$ be a solution obtained by local search on an input graph $G$, $O$ be a fixed optimum solution, $S' := S \setminus O$, and $O' := O \setminus S$.
  It is known that when $G[S' \cup O']$ is planar, there is an absolute constant $C$ such that the $C/\varepsilon^2$-\textsc{Local Search} $1+\varepsilon$-approximates the problem \cite{Mustafa10,Chan12}, that is for a maximization problem, $\lvert S' \rvert \geqslant (1-\varepsilon)|O'|$, implying $\lvert S \rvert \geqslant (1-\varepsilon) \lvert O \rvert$.
  This gives a PTAS with running time $n^{O(1/\varepsilon^2)}$.
  When $G[S' \cup O']$ is even a forest, then it can be shown, and it is somewhat folklore, that a $C/\varepsilon$-\textsc{Local Search} is sufficient.
\end{proof}

\subsection{Strengthening the inapproximability}\label{sec:improving}

There are two directions to improve the hardness-of-approximation results of \cref{sec:triangle-free,sec:high-girth}.
As already mentioned, one can try to match upper and lower bounds in the approximation ratio, or at least to increase the exponent $\delta$ such that an $n^\delta$-approximation would contradict a standard complexity-theoretic assumption.
For triangle-free graphs, for instance, we do \emph{not} expect a matching $n^{1/6}$-approximation.
And a likely outcome is that, ignoring logarithmic factors, the $\sqrt n$-approximation is best possible.
We will actually show that an $n^{1/4 - \varepsilon}$-approximation is unlikely.
The other direction is to derandomize our reductions.
That way the inapproximability would be subject to the more (arguably the most) standard complexity assumption that P is not equal to NP.
Derandomizing without degrading the quality of the gap seems challenging.
We now encapsulate the reductions of \cref{sec:triangle-free,sec:high-girth} so that both improving tasks boil down to exhibiting a randomized or deterministic family of graphs.

We say that an infinite family of graphs $\mathcal C$ is \emph{non-disappearing} if there is a constant $K \in (0,1]$ such that for every positive integer $n$, there is a graph $G \in \mathcal C$ with at least $Kn$ and at most $n/K$ vertices.
  A non-disappearing family is called \emph{efficient} if there is a polynomial-time algorithm which given an integer $n$ (encoded in unary), outputs such a graph $G$.
  For example, a family containing at least one graph for every number of vertices is non-disappearing.
  We denote by $\mathcal G_\gamma$ the set of all graphs with girth at least $\gamma$.
  
\begin{theorem}\label{thm:black-box-deterministic}
  Let $\gamma > 3$ be an integer, $\delta \in (0,1)$ be a real (allowed to depend on $\gamma$), and $\mathcal C$ be an efficient non-disappearing family included in $\mathcal G_\gamma$, such that for every $G \in \mathcal C$ there is no disjoint pair of sets $A, B \subseteq V(G)$ satisfying both $\lvert A \rvert = \lvert B \rvert \geqslant \lvert V(G) \rvert ^{\delta}$ and $E(A,B) = \emptyset$.
  Then \mis in $\mathcal G_{\gamma}$ cannot be $n^{\frac{1-\delta}{2}-\varepsilon}$-approximated, unless P $=$ NP. 
\end{theorem}
\begin{proof}
  We assume all the preconditions hold and follow the construction of \cref{thm:triangle-free,thm:high-girth}.
  We again draw a graph $F$ from graphs of size $N$ and gap $N^{1-\varepsilon}$.
  We substitute every vertex $v$ by an independent set $I(v)$ of size between $KN^{\frac{1+\delta}{1-\delta}}$ and $\frac{1}{K} N^{\frac{1+\delta}{1-\delta}}$ such that there is a $G \in \mathcal C$ of the same size as the obtained graph $G'$, and the sets $I(v)$ are balanced (their size differs by at most 1). 
  Both graphs have $\Theta(N^{\frac{1+\delta}{1-\delta}+1})=\Theta(N^{\frac{2}{1-\delta}})$ vertices, say $c N^{\frac{2}{1-\delta}}$.
  We arbitrary identify the vertices of $G$ and $G'$ in a one-to-one mapping.
  We keep an edge between two vertices $u$ and $v$ if $uv$ is both an edge in $G$ and $G'$.
  Thus we do the ``intersection'' of $G$ and $G'$.
  We call $J$ the final result.
  
  Since $G$ has girth at least $\gamma$, $J$ has also girth at least $\gamma$.
  By assumption on $\mathcal C$, if there is an edge between $u$ and $v$ in $F$, then for every $A \subset I(u)$ and $B \subset I(v)$ both of size $c^\delta N^{\frac{2 \delta}{1-\delta}}$, there is at least one edge in $E_J(A,B)$.
  We observe that $N^{\frac{2 \delta}{1-\delta}} = N^{\frac{1 + \delta}{1-\delta}-1}$ which is, up to constant multiplicative factors, the size of an $I(w)$ divided by $N$.
  Therefore we have the same important property as in \cref{thm:triangle-free,thm:high-girth}.
  Thus we can finish the proof similarly, and conclude that distinguishing between instances with independence number at most $N^{\frac{1+\delta}{1-\delta}+\varepsilon}$ or at least $N^{\frac{2}{1-\delta}-\varepsilon}$ is NP-hard, for an arbitrary small $\varepsilon > 0$.
  The gap is $N^{1-\varepsilon}$ and $n := |V(J)| = \Theta(N^{\frac{2}{1-\delta}})$, hence a gap of $n^{\frac{1-\delta}{2}-\varepsilon}$.
\end{proof}

We now give a randomized counterpart of the previous theorem.
For any integer $\gamma > 3$ and real $\delta \in (0,1)$, we say that a distribution of graphs $\mathcal D$ is \emph{$(\gamma,\delta)$-appropriate} if there is a constant $K \in (0,1]$ and a polynomial-time algorithm, that given an integer $n$ (encoded in unary), draws a graph $G$ of size at least $Kn$ and at most $n/K$ out of this distribution such that with high probability, $G$ has girth at least $\gamma$ and no disjoint pair of sets $A, B \subseteq V(G)$ satisfies both $\lvert A \rvert = \lvert B \rvert \geqslant \lvert V(G) \rvert ^{\delta}$ and $E(A,B) = \emptyset$.

\begin{theorem}\label{thm:black-box-randomized}
  Let $\gamma > 3$ be an integer, $\delta \in (0,1)$ be a real (allowed to depend on $\gamma$), and $\mathcal D$ be a $(\gamma,\delta)$-appropriate distribution.
  Then \mis in $\mathcal G_{\gamma}$ cannot be $n^{\frac{1-\delta}{2}-\varepsilon}$-approximated, unless NP $\subseteq$ BPP. 
\end{theorem}
\begin{proof}
  The proof is the same as \cref{thm:black-box-deterministic}, using a graph drawn from the distribution $\mathcal D$ instead of a deterministic one from $\mathcal C$.
  Therefore we need the stronger assumption that NP is not contained in BPP.
\end{proof}

There are many constructions, all randomized, of triangle-free graphs with smallest possible independence number $\tilde{O}(\sqrt n)$ \cite{Erdos61,Krivelevich95,Erdos95,Kim95,Bollobas11}.
These constructions all follow a simple scheme of starting from the empty graph, ordering the edges of the clique $K_n$, and then inserting an edge if it does not create a triangle, either among the inserted edges or among all the previous edges.
The real difficulty is in the analysis of this probabilistic experiment.
The logarithmic or constant factors were improved and the proofs simplified until Kim obtained a matching bound of $O(\sqrt{n \log n})$~\cite{Kim95}.
This can be seen as the lower bound of $\Omega(n^2/\log n)$ for the off-diagonal Ramsey number $R(3,n)$, matching the upper bound $O(n^2/\log n)$ of Ajtai et al.~\cite{Ajtai80}.

To apply \cref{thm:black-box-randomized}, we would need to check that the triangle-free graphs built in the aforementioned papers do not contain the complement of a large biclique $K_{n^\delta,n^\delta}$.
As, for our purposes, we do not need the optimal bound of Kim, we follow the original proof of Erd\H{o}s~\cite{Erdos61} giving the bound of $O(\sqrt{n} \log n)$.
Going through all the lemmas and replacing occurrences of $K_x$, where $x = O(\sqrt{n} \log n)$, by $K_{x,x}$, the desired result can be obtained.
In our language, the process described in the previous paragraph yields a $(4,1/2)$-appropriate distribution.
This together with \cref{thm:black-box-randomized} improves the inapproximability of \cref{thm:triangle-free}.

\begin{corollary}\label{cor:erdos}
  \mis in $\mathcal G_4$ (i.e., triangle-free graphs) cannot be $n^{\frac{1}{4}-\varepsilon}$-approximated, unless NP $\subseteq$ BPP. 
\end{corollary}

We are not aware of any explicit deterministic construction of triangle-free graphs whose complements do not contain $K_{n^{2/3},n^{2/3}}$ as a subgraph (which would derandomize \cref{thm:triangle-free}), let alone, $K_{\sqrt{n},\sqrt{n}}$.
Deterministic constructions of graphs with large girth and large chromatic number, such as Ramanujan graphs with non-constant degree, might give some lower bound via \cref{thm:black-box-deterministic}, but not as good as \cref{cor:erdos}.
Actually, being based on a tight construction, the inapproximability of \cref{cor:erdos} can only be improved via a totally different route.
One should also not completely rule out that there is an $n^{1/4}$-approximation for \smis on triangle-free graphs.

\section{Concluding remarks}\label{sec:perspectives}

The Erd\H{o}s-Hajnal conjecture has proven particularly difficult.
For example, the cases of $P_5$-free or $C_5$-free graphs are both wide open.
For the few graphs $H$ for which a proof that the Erd\H{o}s-Hajnal property holds, it appears that the proof comes with an efficient algorithm reporting a sufficiently large independent set or clique.
This is what we called the constructive Erd\H{o}s-Hajnal property.
We proposed a first and more humble step (see Theorem~\ref{thm:erdoshajnal-implies-approx}) in proving that a graph $H$ has the constructive Erd\H{o}s-Hajnal property: show that \mis in $H$-free graphs can be approximated within ratio $n^{1-\varepsilon}$ for an $\varepsilon > 0$, an unachievable ratio in general graphs.
As mentioned in the introduction, this is strictly simpler than Erd\H{o}s-Hajnal considering the case of $P_5$.
Yet it does not seem to us that this weaker conjecture is that much simpler now considering the graph $C_5$. 
We believe that efforts to settle the improved approximation conjecture might turn out useful to make progress on the Erd\H{o}s-Hajnal conjecture.
In general, a cross-fertilization between Approximability Theory and the study of favorable Ramsey properties may prove fruitful.
In particular, obtaining an $n^{0.99}$-approximation algorithm for \mis in $C_5$-free seems like a challenging open question.

Of course, classifying the approximability of \textsc{Maximum Independent Set} in $H$-free graphs is also an interesting task by its own means.
On the one hand, already known reductions rule out PTASes in most $H$-free graphs classes, namely for any connected $H$ different from a path or a subdivision of a claw.
On the other hand, a constant-approximation algorithm can be turned into a PTAS in many $H$-free classes, by running the approximation on the input graph elevated to some appropriate power (using for instance the lexicographic product).
This trick, originally used to rule out approximation algorithms for \textsc{Max Clique} in general graphs \cite{Feige91}, works in the setting of $H$-free classes when $H$ satisfies some properties, such as being a prime graph (i.e., having no non-trivial module).
Hence, although \mis admits a constant-factor approximation in $K_{1,t}$-free graphs for any $t \in \mathbb{N}$ (as mentioned in Section~\ref{sec:someH}), it is not in APX when the forbidden graph is a simple tree, such as the 1-subdivision of $K_{1,4}$.
Finally, another interesting consequence of the previous observation concerns $P_t$-free graphs: any constant-factor approximation for \mis in $P_t$-free graphs implies a PTAS (notice that the current ``best'' approximation algorithm in $P_t$-free graphs is a quasi-polynomial approximation scheme~\cite{Chudnovsky19}).




\medskip
\textbf{Acknowledgment.}
We would like to thank Colin Geniet for pointing out to us the remark on graph products mentioned in the conclusion.


\end{document}